\documentclass[twocolumn,english]{IEEEtran}
\usepackage[T1]{fontenc}
\usepackage[latin9]{inputenc}
\usepackage{bm}
\usepackage{amsmath}
\usepackage{graphicx}
\usepackage{amssymb}

\makeatletter
\newtheorem{lemma}{Lemma}
\newtheorem{thm}{Theorem}
\newtheorem{remrk}{Remark}
\newtheorem{cor}{Corollary}

\usepackage{bm}
\usepackage{bbm}
\usepackage{mathrsfs}

\makeatother

\usepackage{babel}

\makeatother

\usepackage{babel}

\begin{document}

\title{A Geometric Approach to Low-Rank Matrix Completion}

\author{Wei Dai$^{*}$, Ely Kerman$^{**}$, Olgica Milenkovic$^{*}$\\
 $^{*}$Department of Electrical and Computer Engineering, $^{**}$Department
of Mathematics \\
 University of Illinois at Urbana-Champaign\\
 Email: \{weidai07,ekerman,milenkov\}@illinois.edu}
\maketitle
\begin{abstract}
The low-rank matrix completion problem can be succinctly stated as
follows: given a subset of the entries of a matrix, find a low-rank
matrix consistent with the observations. While several low-complexity
algorithms for matrix completion have been proposed so far, it remains
an open problem to devise search procedures with provable performance
guarantees for a broad class of matrix models. The standard approach
to the problem, which involves the minimization of an objective function
defined using the Frobenius metric, has inherent difficulties: the
objective function is not continuous and the solution set is not closed.
To address this problem, we consider an optimization procedure that
searches for a column (or row) space that is \emph{geometrically}
consistent with the partial observations. The geometric objective
function is continuous everywhere and the solution set is the closure
of the solution set of the Frobenius metric. We also preclude the
existence of local minimizers, and hence establish strong performance
guarantees, for special completion scenarios, which do not require
matrix incoherence or large matrix size. 
\end{abstract}

\section{Introduction}

In many practical applications of data acquisition, the signals of
interest have a sparse representation in some basis. That is, they
can be well approximated using only a few basis elements. This allows
for efficient sampling and reconstruction of signals \cite{Donoho_IT2006_CompressedSensing,Candes_Tao_IT2006_Robust_Uncertainty_Principles,Candes_Tao_IT2005_decoding_linear_programming,Recht2007_MatrixReconstruction,Candes_Recht_2008_matrix_completion,candes_tao_power_2009}.
More precisely, the number of linear measurements required to capture
a sparse signal can be much smaller than the number of inherent dimensions
of the signal, and various polynomial time algorithms are known for
accurately reconstructing the sparse signal based on these linear
measurements. Due to the significant reduction in sampling resources
and modest requirements for computational resources, sparse signal
processing has been studied intensively \cite{Donoho_IT2006_CompressedSensing,Candes_Tao_IT2006_Robust_Uncertainty_Principles,Candes_Tao_IT2005_decoding_linear_programming,Recht2007_MatrixReconstruction,Candes_Recht_2008_matrix_completion,candes_tao_power_2009}.

There are two categories of sparse signals which frequently arise
in applications. In the first category, the sparse signal can be modeled
a vector with only a small fraction of non-zero entries. Compressive
sensing is the framework of sampling and recovering such signals.
In the second category, the signals are represented by matrices whose
ranks are much smaller than either of their dimensions. In the second
setting, one of the fundamental problems of sparse signal processing
is the low-rank matrix completion problem -- to determine when and
how one can recover a low-rank matrix based on only a subset of its
entries \cite{Candes_Recht_2008_matrix_completion,candes_tao_power_2009,Chandrasekaran2009_incoherence_matrix_decomposition}.

Scores of methods and algorithms have been proposed for low-rank matrix
completion. Many of them are based on similarities between compressive
sensing reconstruction and low-rank matrix completion. In general,
both reconstruction tasks are ill-posed and computationally intractable.
Nevertheless, exact recovery in an efficient manner is possible for
both signal categories provided that the signal is sufficiently sparse
or sufficiently densely sampled. Casting the sparse signal recovery
problem as an optimization problem, $\ell_{1}$-minimization has been
proposed for compressive sensing signal reconstruction \cite{Donoho_IT2006_CompressedSensing,Candes_Tao_IT2006_Robust_Uncertainty_Principles,Candes_Tao_IT2005_decoding_linear_programming}.
Following the same idea, methods based on nuclear norm minimization
have been developed for low-rank matrix completion \cite{Candes_Recht_2008_matrix_completion,candes_tao_power_2009,candes_matrix_noise_2009,cai_singular_2008}.
In terms of greedy algorithms, many of the approaches for low-rank
completion can be viewed as generalizations of their counterparts
for compressive sensing reconstruction. In particular, the ADMiRA
algorithm \cite{lee_Bresler_admira:_2009} is a counterpart of the
subspace pursuit (SP) \cite{Dai_2009_Subspace_Pursuit} and CoSaMP
\cite{Tropp_2009_CoSaMP} algorithms, while the singular value projection
(SVP) method \cite{Meka2009_SVP} extends the iterative hard thresholding
(IHT) \cite{Blumensath_Davies_2009_IHT} approach. There are also
other approaches that utilize some special structural properties of
the low-rank matrices. Examples include the power factorization algorithm
\cite{Haldar_Hernando_powerfactorization_2009}, the OptSpace algorithm
\cite{montanari_keshavan_matrix_2009}, and the subspace evolution
and transfer algorithm \cite{Dai2010_ICASSP_SET}. 

Nevertheless, there is a fundamental problem in low-rank matrix completion
which has not been successfully addressed yet: how to search for a
low-rank matrix consistent with partial observations. The fundamental
difference between compressive sensing and low-rank matrix completion
lies in the knowledge of the {}``sparse basis''. In compressive
sensing, the basis under which the signal is sparse is known a priori.
In principle, the support set of the nonzero entries can be found
by exhaustive search. However, in low-rank matrix completion, the
corresponding {}``sparse basis'' is not known. Note that the set
of all possible bases forms a continuous space. In such a space, {}``exhaustive''
search is impossible. Moreover, we shall show, in Example 1 of Section
\ref{sec:Geometric-Metric}, that a direct gradient-descent search
does not work either. 

The understanding of the search for consistent matrices is incomplete.
There are two special cases where specially designed algorithms can
guarantee a consistent low-rank solution. The first case is when the
low-rank matrix is fully sampled. The consistent low-rank solution
is simply the observation matrix itself. The corresponding {}``sparse
basis'' (singular vectors) can be easily obtained by a singular value
decomposition. The other case is when the rank equals to one. Given
an arbitrary sampling pattern, one simply looks at the ratios between
the revealed entries in the same column and uses these ratios to construct
a column vector that represents the column space. This method is guaranteed
to find a consistent solution for rank-one matrices. However, it remains
an open problem how to extend this method for general ranks. Hence,
such an approach is not universal. On the other hand, none of existing
general algorithms provides performance guarantee even for the rank-one
case. The performance guarantee of nuclear norm minimization is built
on incoherence conditions, which only holds with high probability
when the low-rank matrix is drawn randomly from certain ensembles
and when the size of the matrix is sufficiently large. Our understanding
of low-rank matrix completion is far from complete. 

Our approach to address these issues is summarized as follows.
\begin{enumerate}
\item We provide a framework for searching for a low-rank matrix that is
\emph{consistent} with the partial observations. There is no requirement
that such a matrix is unique: if there is a unique low-rank solution,
we should be able to find this unique matrix; otherwise, it suffices
to find just one solution that agrees with the revealed entries. In
our approach, we assume that the rank of the underlying low-rank matrix
is known a priori. Finding a consistent low-rank matrix is equivalent
to finding a consistent column/row space. This is different from the
OptSpace algorithm in \cite{montanari_keshavan_matrix_2009}, where
the search is performed on both column and row spaces simultaneously. 
\item We propose a geometric performance metric to measure the consistency
between the estimated column space and the partial observations. In
the literature, the standard approach is to minimize an objective
function that is defined via the Frobenius norm. As we shall illustrate
with explicit examples, this objective function may have singularities,
and therefore the corresponding solution set may not be closed. Hence,
we introduce a new formulation where consistency is now defined in
geometric terms. This allows us to address the difficulties related
to the Frobenius metric. In particular, we show that our geometric
objective function is always continuous. The set of the corresponding
consistent solutions is the closure of the set corresponding to the
Frobenius norm. This new metric allows for provably strong performance
guarantees, described below. 
\item We provide strong performance guarantees for special completion scenarios:
rank-one matrices with arbitrary sampling patterns, and fully sampled
matrices%
\footnote{For full sampled matrices, even though using a simple singular value
decomposition produces a consistent column space, it is not clear
that a randomlly initialized search would converge to a consistent
column space. In what follows, we prove that this is the case. %
} of arbitrary rank. For these two scenarios, a gradient descent search
starting from a random point will converge to a global minimum with
probability one. \emph{More importantly, if the partial observations
admit a unique consistent solution, this search procedure finds this
unique solution with probability one.} The performance guarantees
are different from those previously established in literature. Roughly
speaking, previous performance guarantees require large matrix sizes
and only hold with high probability. Ours hold with probability one
regardless of matrix size. It is also worth noting that we do not
require incoherence conditions, which are essential for the performance
guarantees of nuclear norm minimization. Unfortunately, we are presently
unable to obtain performance guarantees for more general cases.
\end{enumerate}
\vspace{0in}

The paper is organized as follows. In Section~\ref{sec:Low-Rank-Matrix-Completion}
we introduce the low-rank matrix completion problem, and some background
material regarding Grassmann manifolds and their geometry. In Section~\ref{sec:Geometric-Metric}
we show that formulating the low-rank matrix completion problem as
an optimization problem using the Frobenius norm may yield singularities
which can obstruct standard minimization algorithms. We then propose
a new geometric formulation of the problem as a remedy to this difficulty.
This new formulation allows for strong performance guarantees that
are presented in Section~\ref{sec:Performance-Guarantee}. Section
\ref{sec:Conclusion} summarizes the main contributions of the work.
Proofs of the main results are presented in the Appendices.

\section{\label{sec:Low-Rank-Matrix-Completion}Low-Rank Matrix Completion
and Preliminaries}

Let $\bm{X}\in\mathbb{R}^{m\times n}$ be an unknown matrix with rank
$r\leq\min\left(m,n\right)$, and let $\Omega\subset\left[m\right]\times\left[n\right]$
be the set of indices of the observed entries, where $\left[K\right]=\left\{ 1,2,\cdots,K\right\} $.
Define the projection operator $\mathcal{P}_{\Omega}$ by \begin{align*}
\mathcal{P}_{\Omega}:\;\mathbb{R}^{m\times n} & \rightarrow\mathbb{R}^{m\times n}\\
\mathcal{P}_{\Omega}(\bm{X}) & \mapsto\bm{X}_{\Omega},\;\mbox{where }\left(\bm{X}_{\Omega}\right)_{i,j}=\begin{cases}
\bm{X}_{i,j} & \mbox{if }\left(i,j\right)\in\Omega\\
0 & \mbox{if }\left(i,j\right)\notin\Omega\end{cases}.\end{align*}
 The \emph{consistent matrix completion} problem is to find \emph{one}
rank-$r$ matrix $\bm{X}^{\prime}$ that is consistent with the observations
$\bm{X}_{\Omega}$, i.e., \begin{align}
\left(P0\right):\; & \mbox{find a }\bm{X}^{\prime}\mbox{ such that }\nonumber \\
 & \mbox{rank}\left(\bm{X}^{\prime}\right)=r\mbox{ and }\mathcal{P}_{\Omega}\left(\bm{X}^{\prime}\right)=\mathcal{P}_{\Omega}\left(\bm{X}\right)=\bm{X}_{\Omega}.\label{eq:P0}\end{align}
 By definition, this problem is well defined since $\bm{X}_{\Omega}$
is obtained from some rank-$r$ matrix $\bm{X}$ which is therefore
a solution. As in other works, \cite{lee_Bresler_admira:_2009,Haldar_Hernando_powerfactorization_2009,montanari_keshavan_matrix_2009},
we assume that the rank $r$ is given. In practice, one may try to
sequentially guess a rank bound until a satisfactory solution has
been found.

We also introduce the (standard) projection operator $\mathcal{P}$,
\begin{align*}
\mathcal{P}:\;\mathbb{R}^{m}\times\mathbb{R}^{m\times k} & \rightarrow\mathbb{R}^{m}\\
\mathcal{P}\left(\bm{x},\bm{U}\right) & \mapsto\bm{y}=\bm{U}\bm{U}^{\dagger}\bm{x},\end{align*}
 where $1\le k\le m$, and where the superscript $\dagger$ denotes
the pseudoinverse of a matrix. Let $\mbox{span}\left(\bm{U}\right)$
denote the subspace spanned by the columns of the matrix $\bm{U}$,
i.e.,\[
\mbox{span}\left(\bm{U}\right)=\left\{ \bm{v}\in\mathbb{R}^{m}:\;\bm{v}=\bm{U}\bm{w}\mbox{ for some }\bm{w}\in\mathbb{R}^{m}\right\} .\]
 One can describe $\mathcal{P}\left(\bm{x},\bm{U}\right)$, in geometric
terms, as the projection of the vector $\bm{x}$ onto $\mbox{span}\left(\bm{U}\right)$.
It should be observed that $\bm{U}^{\dagger}\bm{x}$ is the global
minimizer of the quadratic optimization problem $\min_{\bm{w}\in\mathbb{R}^{k}}\;\left\Vert \bm{x}-\bm{U}\bm{w}\right\Vert _{2}^{2}.$

\subsection{Search for a consistent column space}

We now show that the problem $\left(P0\right)$ is equivalent to finding
a column space consistent with the observed entries of $\bm{X}$.

Let $\mathcal{U}_{m,r}$ be the set of $m\times r$ matrices with
$r$ orthonormal columns, i.e., $\mathcal{U}_{m,r}=\left\{ \bm{U}\in\mathbb{R}^{m\times r}:\;\bm{U}^{T}\bm{U}=\bm{I}_{r}\right\} .$
Define the function $f_{F}:\;\mathcal{U}_{m,r}\rightarrow\mathbb{R}$
by setting \begin{equation}
f_{F}(\bm{U})=\underset{\bm{W}\in\mathbb{R}^{n\times r}}{\min}\left\Vert \bm{X}_{\Omega}-\mathcal{P}_{\Omega}\left(\bm{U}\bm{W}^{T}\right)\right\Vert _{F}^{2},\label{eq:f_U}\end{equation}
 where $\left\Vert \cdot\right\Vert _{F}$ denotes the Frobenius norm.
This function measures the consistency between the matrix $\bm{U}$
and the observations $\bm{X}_{\Omega}$. In particular, if $f_{F}\left(\bm{U}\right)=0$,
then there exists a matrix $\bm{W}$ such that the rank-$r$ matrix
$\bm{U}\bm{W}^{T}$ satisfies $\mathcal{P}_{\Omega}\left(\bm{U}\bm{W}^{T}\right)=\bm{X}_{\Omega}$.
Hence, the consistent matrix completion problem is equivalent to \begin{align}
\left(P1\right):\; & \mbox{find }\bm{U}\in\mathcal{U}_{m,r}\;\mbox{such that }f_{F}\left(\bm{U}\right)=0.\label{eq:P1}\end{align}

In fact, $f_{F}(U)$ depends only on the subspace $\mbox{span}\left(\bm{U}\right)$
since the columns of a matrix of the form $\bm{U}\bm{W}^{T}$ all
lie in $\mbox{span}\left(\bm{U}\right)$. Hence, to solve the consistent
matrix completion problem, it suffices to find a \emph{column space}
that is consistent with the observed entries. Note that the same conclusion
holds for the row space as well. For simplicity, we restrict our attention
to the column space only.

\subsection{\label{sub:Grassmann-Manifold}Grassmann Manifolds}

The set of column spaces of elements in $\mathcal{U}_{m,r}$ can be
identified with the Grassmann manifold $\mathcal{G}_{m,r}$, the set
of $r$-dimensional subspaces in the $m$-dimensional Euclidean space
$\mathbb{R}^{m}$. This is a smooth compact manifold of dimension
$r(m-r)$. Conversely, every element, say $\mathscr{U}\in\mathcal{G}_{m,r}$
can be presented by a \emph{generator matrix} $\bm{U}\in\mathcal{U}_{m,r}$
satisfying $\mathrm{span}\left(\bm{U}\right)=\mathscr{U}$. However,
this presentation of $\mathscr{U}$ by a generator matrix is clearly
not unique. Nevertheless, it follows from the discussion in the previous
section that the function $f_{F}$ descends to a function on $\mathcal{G}_{m,r}$.
Thus, problem $\left(P1\right)$ can be viewed as an optimization
problem on the compact manifold $\mathcal{G}_{m,r}$.

In this section we recall some facts concerning the geometry of Grassmann
manifolds which will be useful in addressing this and similar optimization
problems. For the proofs of these facts the reader is referred to
\cite{edelman_optimization_manifolds_1998}. We begin by recalling
the construction of the standard Riemannian metric, $g_{m,r}$, on
$\mathcal{G}_{m,r}$. Note that the group $\mathcal{U}_{m,m}$ of
orthogonal $m\times m$ matrices acts transitively on $\mathcal{G}_{m,r}$
(by multiplication on generator matrices). More precisely, $\mathcal{G}_{m,r}$
can be described as a quotient of $\mathcal{U}_{m,m}$, i.e., \[
\mathcal{G}_{m,r}=\mathcal{U}_{m,m}/(\mathcal{U}_{m-r,m-r}\times\mathcal{U}_{r,r})\]
 Now, as a compact Lie group, $\mathcal{U}_{m,m}$ has a standard
(bi-invariant) Riemannian metric (can be defined by using inner product
in the tangent space). This descends to the quotient $\mathcal{G}_{m,r}$
as the metric $g_{m,r}$. By construction, $g_{m,r}$ is invariant
under the action of $\mathcal{U}_{m,m}$.

The metric $g_{m,r}$ determines a chordal distance function and geodesic
curves on $\mathcal{G}_{m,r}$ which will play an important role in
what follows. To obtain the relevant formulas for these objects we
require the notion of the \emph{principal angles} between two subspaces
\cite{Conway_96_PackingLinesPlanes,Dai2008_small_ball_quantization_GM}.
Consider the subspaces $\mbox{span}\left(\bm{U}\right)$ and $\mbox{span}\left(\bm{V}\right)$
of $\mathbb{R}^{m}$ for some $\bm{U}\in\mathcal{U}_{m,p}$ and $\bm{V}\in\mathcal{U}_{m,q}$.
The principal angles between these two subspaces can be defined in
the following constructive manner. Without loss of generality, assume
that  $1\le p\le q\le m$. Let $\bm{u}_{1}\in\mbox{span}\left(\bm{U}\right)$
and $\bm{v}_{1}\in\mbox{span}\left(\bm{V}\right)$ be unit-length
vectors such that $\left|\bm{u}_{1}^{T}\bm{v}_{1}\right|$ is maximal.
Inductively, let $\bm{u}_{k}\in\mbox{span}\left(\bm{U}\right)$ and
$\bm{v}_{k}\in\mbox{span}\left(\bm{V}\right)$ be unit vectors such
that $\bm{u}_{k}^{T}\bm{u}_{j}=0$ and $\bm{v}_{k}^{T}\bm{v}_{j}=0$
for all $1\le j<k$ and $\left|\bm{u}_{k}^{T}\bm{v}_{k}\right|$ is
maximal. The principal angles are then defined as \[
\alpha_{k}=\arccos\bm{u}_{k}^{T}\bm{v}_{k}\]
 for $k=1,2,\cdots,p$.

Alternatively, the principal angles can be computed via singular value
decomposition. Consider the singular value decomposition $\bm{U}\bm{U}^{T}\bm{V}\bm{V}^{T}=\bar{\bm{U}}\bm{\Lambda}\bar{\bm{V}}^{T}$,
where $\bar{\bm{U}}\in\mathcal{U}_{m,p}$ and $\bar{\bm{V}}\in\mathcal{U}_{m,p}$
contain the first $p$ left and right singular vectors, respectively,
and $\bm{\Lambda}\in\mathbb{R}^{p\times p}$ is a diagonal matrix
comprised of singular values $\lambda_{1}\ge\cdots\ge\lambda_{p}$.
Then the $k^{th}$ columns of $\bar{\bm{U}}$ and $\bar{\bm{V}}$
correspond to the vectors $\bm{u}_{k}$ and $\bm{v}_{k}$ used in
the constructive definition, respectively. The $k^{th}$ singular
value $\lambda_{k}$ defines the $k^{th}$ principal angle $\alpha_{k}$
via \[
\cos\alpha_{k}=\lambda_{k}.\]

\medskip{}

\noindent \textbf{Chordal distance on $\mathcal{G}_{m,r}$.} For $\bm{U}_{1}$
and $\bm{U}_{2}$ in $\mathcal{U}_{m,r}$, the \emph{chordal distance}
between the two subspaces $\mbox{span}\left(\bm{U}_{1}\right)$ and
$\mbox{span}\left(\bm{U}_{2}\right)$ in $\mathcal{G}_{m,r}$ is given,
in terms of the $p$ principal angles between them, via the formula
\[
\sqrt{\sum_{k=1}^{r}\sin^{2}\alpha_{k}}.\]
The chordal distance can also be expressed in terms of singular values
as \[
\sqrt{\sum_{k=1}^{r}(1-\lambda_{k}^{2})}.\]

\medskip{}

\noindent \textbf{Geodesics on $\mathcal{G}_{m,r}$.} We will use
the gradient descent method on $\mathcal{G}_{m,r}$ to search for
consistent column spaces. This will require some information concerning
the geodesics of the metric $g_{m,r}$ on $\mathcal{G}_{m,r}$ which
we now recall. 

Roughly speaking, a geodesic in a manifold is a generalization of
the notion of a straight line in the Euclidean space: given any two
points in $\mathcal{G}_{m,r}$, among all curves that connect these
two points, the one of the shortest length is geodesic. More precisely,
fix a subspace $\mathscr{U}$ in $\mathcal{G}_{m,r}$ and a tangent
vector $\mathscr{H}$ to $\mathcal{G}_{m,r}$ at $\mathscr{U}$. Let
$\bm{U}\in\mathcal{U}_{m,r}$ be a generator matrix for $\mathscr{U}$.
The tangent space to $\mathcal{G}_{m,r}$ at $\mathscr{U}$ can be
identified with the set of \emph{horizontal} tangent vectors to $\bm{U}$,
i.e., the set of tangent vectors $\bm{W}$ at $\bm{U}$ which satisfy
$\bm{U}^{T}\bm{W}=0$ \cite{edelman_optimization_manifolds_1998}.
Let $\bm{H}\in\mathbb{R}^{m\times r}$ be the horizontal tangent vector
at $\bm{U}$ which corresponds to $\mathscr{H}$ and set \begin{equation}
\bm{U}\left(t\right)=\left[\bm{U}\bm{V}_{H},\bm{U}_{H}\right]\left[\begin{array}{c}
\cos\left(\bm{S}_{H}t\right)\\
\sin\left(\bm{S}_{H}t\right)\end{array}\right]\bm{V}_{H}^{T},\label{eq:geodesic}\end{equation}
 where $\bm{U}_{H}\bm{S}_{H}\bm{V}_{H}^{T}$ is the compact singular
value decomposition of $\bm{H}$. Then $\mathrm{span}\left(\bm{U}\left(t\right)\right)$
is the unique geodesic of $g_{m,r}$ which starts at $\mathscr{U}$
with {}``initial velocity'' $\mathscr{H}$.

We now use this general solution for the geodesic flow of $g_{m,r}$
to establish the following technical result concerning geodesics between
a given pair of subspaces. 
\begin{lemma}
\label{lem:Grassmann-Manifold} Fix two elements $\bm{U}_{1}$ and
$\bm{U}_{2}$ of $\mathcal{U}_{m,r}$. Let $\bm{V}_{1}\bm{\Lambda}\bm{V}_{2}^{T}$
be the singular value decomposition of the matrix $\bm{U}_{1}^{T}\bm{U}_{2}$,
and denote the $i^{th}$ singular value by $\lambda_{i}=\cos\alpha_{i}$.
Set $\bar{\bm{U}}_{1}=\bm{U}_{1}\bm{V}_{1}$ and $\bar{\bm{U}}_{2}=\bm{U}_{2}\bm{V}_{2}$
and note that $\bar{\bm{U}}_{1}^{T}\bar{\bm{U}}_{2}=\bm{\Lambda}.$ 
\begin{enumerate}
\item Consider the path \begin{equation}
\bm{U}\left(t\right)=\left[\bar{\bm{U}}_{1},\bm{G}\right]\left[\begin{array}{c}
\mbox{diag}\left(\left[\cdots,\cos\alpha_{i}t,\cdots\right]\right)\\
\mbox{diag}\left(\left[\cdots,\sin\alpha_{i}t,\cdots\right]\right)\end{array}\right]\bm{V}_{1}^{T},\label{eq:geo-path}\end{equation}
 where the columns of $\bm{G}=\left[\cdots,\bm{g}_{i},\cdots\right]\in\mathbb{R}^{m\times r}$
are defined as follows \[
\bm{g}_{i}=\begin{cases}
\frac{\bar{\bm{U}}_{2,:i}-\lambda_{i}\bar{\bm{U}}_{1,:i}}{\left\Vert \bar{\bm{U}}_{2,:i}-\lambda_{i}\bar{\bm{U}}_{1,:i}\right\Vert } & \;\mbox{if }\lambda_{i}\ne1,\\
\bm{0} & \;\mbox{if }\lambda_{i}=1.\end{cases}\]
Here, the subscript $_{:i}$ denotes the $i^{th}$ column of the corresponding
matrix. Then the path $\mbox{span}\left(\bm{U}\left(t\right)\right)$
is a geodesic of $g_{m,r}$ such that $\mbox{span}\left(\bm{U}\left(0\right)\right)=\mbox{span}\left(\bm{U}_{1}\right)$
and $\mbox{span}\left(\bm{U}\left(1\right)\right)=\mbox{span}\left(\bm{U}_{2}\right)$. 
\item Let $\bar{\bm{x}}\in\mbox{span}\left(\bm{U}_{2}\right)$ be a unit-norm
vector. It's clear that there exists a unique $\bar{\bm{w}}\in\mathcal{U}_{r,1}$
such that $\bar{\bm{x}}=\bar{\bm{U}}_{2}\bar{\bm{w}}$. Suppose that
$\bar{\bm{x}}\notin\mbox{span}\left(\bar{\bm{U}}_{1}\right)$. Let
$k$ the number of the singular values of $\bar{\bm{U}}_{1}^{T}\bar{\bm{U}}_{2}$
that equal to one. Then $k<r$ and there exists an index $j\in\left[r\right]$
such that $k<j\le r$ and $\bar{w}_{j}\ne0$. 
\end{enumerate}
\end{lemma}
\begin{proof}
Clearly, $\bm{U}\left(0\right)=\bm{U}_{1}$. Since $\bar{\bm{U}}_{1}^{T}\bar{\bm{U}}_{2}=\bm{\Lambda},$
we have \begin{align*}
 & \left\Vert \bar{\bm{U}}_{2,:i}-\lambda_{i}\bar{\bm{U}}_{1,:i}\right\Vert ^{2}\\
 & =1-2\lambda_{i}\left\langle \bar{\bm{U}}_{2,:i},\bar{\bm{U}}_{1,:i}\right\rangle +\lambda_{i}^{2}\\
 & =1-\lambda_{i}^{2}.\end{align*}
Thus, we have \begin{align*}
\bm{U}\left(1\right) & =\left[\cdots,\bar{\bm{U}}_{1,:i}\cos\alpha_{i}+\bm{g}_{i}\sin\alpha_{i},\cdots\right]\bm{V}_{1}^{T}\\
 & =\left[\cdots,\bar{\bm{U}}_{1,:i}\lambda_{i}+\bm{g}_{i}\sqrt{1-\lambda_{i}^{2}},\cdots\right]\bm{V}_{1}^{T}\\
 & =\left[\cdots,\bar{\bm{U}}_{1,:i}\lambda_{i}+\bm{g}_{i}\left\Vert \bar{\bm{U}}_{2,:i}-\lambda_{i}\bar{\bm{U}}_{1,:i}\right\Vert ,\cdots\right]\bm{V}_{1}^{T}\\
 & =\left(\bar{\bm{U}}_{1}\bm{\Lambda}+\left(\bar{\bm{U}}_{2}-\bar{\bm{U}}_{1}\bm{\Lambda}\right)\right)\bm{V}_{1}^{T}\\
 & =\bm{U}_{2}\bm{V}_{2}\bm{V}_{1}^{T}.\end{align*}
 Hence, $\mbox{span}\left(\bm{U}\left(1\right)\right)=\mbox{span}\left(\bm{U}_{2}\right)$.
To prove the first part of the lemma it just remains to show that
$\mbox{span}\left(\bm{U}(t)\right)$ is geodesic. Setting $\bm{H}=\dot{\bm{U}}\left(0\right)$
we have \begin{equation}
\bm{H}=\bm{G}\mbox{diag}\left(\left[\cdots,\alpha_{i},\cdots\right]\right)\bm{V}_{1}^{T}.\label{svH}\end{equation}
 We first verify that the tangent vector $\bm{H}$ is horizontal which
is equivalent to showing that $\bm{U}_{1}^{T}\bm{H}=0$. According
to the definition of the vectors $\bm{g}_{i}$, when $\lambda_{i}\ne1$,
one has \[
\left\Vert \bar{\bm{U}}_{2,:i}-\lambda_{i}\bar{\bm{U}}_{1,:i}\right\Vert \ne0\]
 and \begin{align*}
\bar{\bm{U}}_{1}^{T}\bm{g}_{i} & =\frac{1}{\left\Vert \bar{\bm{U}}_{2,:i}-\lambda_{i}\bar{\bm{U}}_{1,:i}\right\Vert }\bar{\bm{U}}_{1}^{T}\left(\bar{\bm{U}}_{2,:i}-\lambda_{i}\bar{\bm{U}}_{1,:i}\right)\\
 & =\frac{1}{\left\Vert \bar{\bm{U}}_{2,:i}-\lambda_{i}\bar{\bm{U}}_{1,:i}\right\Vert }\lambda_{i}\bm{e}_{i}-\lambda_{i}\bm{e}_{i}=\bm{0}.\end{align*}
 Hence, \[
\bm{U}_{1}^{T}\bm{G}=\bm{V}_{1}^{T}\bar{\bm{U}}_{1}^{T}\bm{G}=\bm{0}.\]
 By \eqref{svH}, this implies that $\bm{U}_{1}^{T}\bm{H}=0$, as
desired. Note that equation \eqref{svH} can also be viewed as an
expression for the compact singular value decomposition of $\bm{H}$.
It then follows directly from \eqref{eq:geodesic} that $\mbox{span}\left(\bm{U}(t)\right)$
is indeed a geodesic.

To prove the second part of the lemma, let $\bm{u}_{1,1},\cdots,\bm{u}_{1,r}$
and $\bm{u}_{2,1},\cdots,\bm{u}_{2,r}$ be the column vectors of the
matrix $\bar{\bm{U}}_{1}$ and $\bar{\bm{U}}_{2}$, respectively.
By assumption, $\lambda_{1}=\cdots=\lambda_{k}=1$ and $1>\lambda_{k+1}\ge\cdots\ge\lambda_{r}$.
Hence, \begin{align*}
 & \bm{u}_{1,j}=\bm{u}_{2,j},\;\mbox{for all }j\le k,\;\mbox{and}\\
 & \left\langle \bm{u}_{1,j},\bm{u}_{2,j}\right\rangle =\lambda_{j}<1,\;\mbox{for all }k<j\le r.\end{align*}
 Suppose that $k=r$. Then \[
\bar{\bm{x}}=\bar{\bm{U}}_{2}\bar{\bm{w}}=\bar{\bm{U}}_{2}\bar{\bm{w}}\in\mbox{span}\left(\bar{\bm{U}}_{1}\right),\]
 which contradicts the assumption that $\bar{\bm{x}}\notin\mbox{span}\left(\bar{\bm{U}}_{1}\right)$.
Hence, we have $k<r$. Now suppose that $\bar{w}_{k+1}=\cdots=\bar{w}_{r}=0$.
Then \[
\bar{\bm{x}}=\sum_{j=1}^{k}\bm{u}_{2,j}\bar{w}_{j}=\sum_{j=1}^{k}\bm{u}_{1,j}\bar{w}_{j}\in\mbox{span}\left(\bm{U}_{1}\right),\]
 which again contradicts the assumption that $\bar{\bm{x}}\notin\mbox{span}\left(\bm{U}_{1}\right)$.
Hence, there exists a $j$ such that $k<j\le r$ and $\bar{w}_{j}\ne0$.
This completes the proof. 
\end{proof}
\medskip{}

\noindent \textbf{An invariant measure on $\mathcal{G}_{m,r}$.} The
space $\mathcal{U}_{m,m}$ admits a standard invariant measure (the
Haar measure) \cite{James_54_Normal_Multivariate_Analysis_Orthogonal_Group}.
This descends to a measure $\mu$ on $\mathcal{G}_{m,r}$ which is
also invariant in the following sense: for any measurable set $\mathcal{M}\subset\mathcal{G}_{m,r}$
and any $\bm{A}\in\mathcal{U}_{m,m}$, one has $\mu\left(\mathcal{M}\right)=\mu\left(\bm{A}\mathcal{M}\right)$,
where $\bm{A}\mathcal{M}=\left\{ \mbox{span}\left(\bm{A}\bm{U}\right):\;\bm{U}\in\mathcal{U}_{m,r},\;\mbox{span}\left(\bm{U}\right)\in\mathcal{M}\right\} $
\cite{James_54_Normal_Multivariate_Analysis_Orthogonal_Group,Dai2008_small_ball_quantization_GM}.
This invariant measure defines the uniform/isotropic distribution
on the Grassmann manifold. Furthermore, let $\mbox{span}\left(\bm{U}\right)\in\mathcal{G}_{m,r}$
be fixed and $\mbox{span}\left(\bm{V}\right)\in\mathcal{G}_{m,r}$
be drawn randomly from the isotropic distribution. The joint probability
density function of the principal angles between the spans of $\bm{U}$
and $\bm{V}$ is explicitly given in \cite{James_54_Normal_Multivariate_Analysis_Orthogonal_Group,Adler_2004_Integrals_Grassmann,Dai2008_small_ball_quantization_GM,Dai_Globecom07_large_balls}.
Two properties of this density function will be relevant to our later
analysis: first, it is independent of the choice of $\bm{U}$; second,
there is no mass point.

\section{\label{sec:Geometric-Metric}From the Frobenius Norm to the Geometric
Metric}

In the previous section, we showed that the matrix completion problem
reduces to a search for a consistent column space. In other words,
one only needs to find a global minimum of the objective function
$f_{F}\left(\bm{U}\right),$ where \begin{equation}
f_{F}\left(\bm{U}\right)\triangleq\underset{\bm{W}\in\mathbb{R}^{r\times n}}{\min}\left\Vert \bm{X}_{\Omega}-\mathcal{P}_{\Omega}\left(\bm{U}\bm{W}\right)\right\Vert _{F}^{2}.\label{eq:objective-fn-Frobenius}\end{equation}
However, as we shall show in Section \ref{sub:Frobenius-Norm-Fails},
this approach has a serious drawback: the objective function (\ref{eq:objective-fn-Frobenius})
is not a continuous function of the variable $\bm{U}$. The discontinuity
of the objective function is due to the composition of the Frobenius
norm with the projection operator $\mathcal{P}_{\Omega}$. It may
prevent gradient-descent-based algorithms from converging to a global
optimum (see \cite{Dai2010_ICASSP_SET} for a detailed example). To
address this issue, we propose another objective function $f_{G}\left(\bm{U}\right)$
based on the geometry of the problem, detailed in Section \ref{sub:Geometric-Metric}.
To solve the matrix completion problem, one then needs to solve the
problem \begin{equation}
\left(P2\right):\;\mbox{find a }\bm{U}\in\mathcal{U}_{m,r}\;\mbox{such that }f_{G}\left(\bm{U}\right)=0.\label{eq:matrix-completion-chordal-dist}\end{equation}
 where $f_{G}$ denotes the geometric metric, which is formally defined
in Section \ref{sub:Geometric-Metric}. 

In the rest of this section, we shall show that the new objective
function $f_{G}$ is a continuous function. Furthermore, we shall
show that the preimage of $f_{G}\left(\bm{U}\right)=0$ is the \emph{closure}
of the preimage of $f_{F}\left(\bm{U}\right)=0$. Because of these
nice properties of the geometric objective function, one can derive
strong performance guarantees for gradient descent methods, as described
in Section \ref{sec:Performance-Guarantee}.

\subsection{\label{sub:Frobenius-Norm-Fails}Why the Frobenius Norm Fails}

We use an example to show that the objective function (\ref{eq:objective-fn-Frobenius})
based on the Frobenius norm is not continuous. Let $\bm{x}_{\Omega,i}$
be the $i^{th}$ column of the matrix $\bm{X}_{\Omega}$. Let $\Omega_{i}\subset\left[m\right]$
be the set of indices of known entries in the $i^{th}$ column. We
use $\mathcal{P}_{\Omega,i}$ to denote the projection operator corresponding
to the index set of $\Omega_{i}$. By additivity of the squared Frobenius
norm, the objective function can be written as a sum of atomic functions,
i.e., \begin{align*}
f_{F}\left(\bm{U}\right) & =\underset{\bm{W}\in\mathbb{R}^{r\times n}}{\min}\left\Vert \bm{X}_{\Omega}-\mathcal{P}_{\Omega}\left(\bm{U}\bm{W}\right)\right\Vert _{F}^{2}\\
 & =\sum_{i=1}^{n}\underbrace{\underset{\bm{w}_{i}\in\mathbb{R}^{r}}{\min}\left\Vert \bm{x}_{\Omega,i}-\mathcal{P}_{\Omega,i}\left(\bm{U}\bm{w}_{i}\right)\right\Vert _{F}^{2}}_{f_{F,i}\left(\bm{U}\right)}.\end{align*}
 Denote the $i^{th}$ \emph{atomic function} by $f_{F,i}\left(\bm{U}\right)$.
It can be verified that \begin{align*}
f_{F,i}\left(\bm{U}\right) & =\underset{\bm{w}\in\mathbb{R}^{r}}{\min}\left\Vert \bm{x}_{\Omega,i}-\mathcal{P}_{\Omega,i}\left(\bm{U}\bm{w}_{i}\right)\right\Vert _{F}^{2}\\
 & =\left\Vert \bm{x}_{\Omega,i}-\mathcal{P}\left(\bm{x}_{\Omega,i},\bm{U}_{\Omega_{i}}\right)\right\Vert _{F}^{2.},\end{align*}
where $\bm{U}_{\Omega_{i}}=\left[\mathcal{P}_{\Omega,i}\left(\bm{u}_{1}\right),\cdots,\mathcal{P}_{\Omega,i}\left(\bm{u}_{r}\right)\right]$
and $\bm{u}_{1},\cdots,\bm{u}_{r}$ are column vectors of the matrix
$\bm{U}$. We show in the next example that an atomic function, say
$f_{F,1}\left(\bm{U}\right)$, may not be continuous. 

\vspace{0.03in}

\emph{Example 1:} \label{Example-1}Suppose that $\bm{x}_{\Omega,1}=\left[0,1,1\right]^{T}$
and $\Omega_{1}=\left\{ 2,3\right\} $. Let $\bm{U}$ be of the form
$\bm{U}=\left[\sqrt{1-2\epsilon^{2}},\epsilon,\epsilon\right]^{T}\in\mathcal{U}_{3,1}$
where $\epsilon\in\left[-1/\sqrt{2},1/\sqrt{2}\right]$. For a given
$\bm{U}$, the atomic function $f_{F,1}\left(\bm{U}\right)$ is given
by \[
f_{F,1}\left(\bm{U}\right)=\underset{w\in\mathbb{R}}{\min}\left\Vert \left[0,1,1\right]^{T}-\mathcal{P}_{\Omega,1}\left(\bm{U}w\right)\right\Vert _{F}^{2}.\]
This is a quadratic optimization problem and can be easily solved.
The optimal $w^{*}$ is given by\[
w^{*}=\begin{cases}
\frac{2}{\epsilon} & \;\mbox{if }\epsilon\ne0,\\
0 & \;\mbox{if }\epsilon=0.\end{cases}\]
 Hence, one has \[
f_{F,1}\left(\bm{U}\left(\epsilon\right)\right)=\begin{cases}
0 & \;\mbox{if }\epsilon\in\left[-\frac{1}{\sqrt{2}},0\right)\bigcup\left(0,\frac{1}{\sqrt{2}}\right],\\
2 & \;\mbox{if }\epsilon=0.\end{cases}\]
which shows that $f_{F,1}\left(\bm{U}\left(\epsilon\right)\right)$
has a singular point at $\epsilon=0$.

\vspace{0.03in}

It is straightforward to verify that the overall objective function
(\ref{eq:objective-fn-Frobenius}) is also a discontinuous function
of $\bm{U}$. As we argued in \cite{Dai2010_ICASSP_SET}, this discontinuity
creates so called barriers, which may prevent gradient-descent algorithms
from converging to a global minimum. Hence, one seeks an optimization
criteria that will allow for a continuous objective function and consequently,
no search path barriers.

\subsection{\label{sub:Geometric-Metric}A Geometric Metric}

To address the problem due to the singularities of the objective functions,
we propose to replace the Frobenius norm by a geometric performance
metric.

In this case, the objective function is defined as \[
f_{G}\left(\bm{U}\right)=\sum_{i=1}^{n}f_{G,i}\left(\bm{U}\right),\]
 where $f_{G,i}\left(\bm{U}\right)$ denotes the geometric metric
corresponding to the $i^{th}$ column, defined as follows. If $\bm{x}_{\Omega,i}=\bm{0}$,
we set $f_{G,i}\left(\bm{U}\right)=0$. Henceforth, we only consider
the case when $\bm{x}_{\Omega,i}\ne\bm{0}$. For any $\bm{x}_{\Omega,i}\ne\bm{0}$,
let $\bar{\bm{x}}_{\Omega,i}=\bm{x}_{\Omega,i}/\left\Vert \bm{x}_{\Omega,i}\right\Vert _{F}$
be the normalized vector $\bm{x}_{\Omega,i}$. Let $\Omega_{i}^{c}=\left\{ 1,2,\cdots,m\right\} \backslash\Omega_{i}$
be the complement of $\Omega_{i}$. Let $\bm{e}_{k}\in\mathbb{R}^{m}$
be the $k^{th}$ natural basis vector, i.e., the $k^{th}$ entry of
$\bm{e}_{k}$ equals to one and all other entries are zero. Define
\begin{equation}
\bm{B}_{i}=\left[\bar{\bm{x}}_{\Omega,i},\bm{e}_{k_{1}},\cdots,\bm{e}_{k_{\ell}}\right],\label{eq:def-Bi}\end{equation}
 where $\left\{ k_{1},\cdots,k_{\ell}\right\} =\Omega_{i}^{c}$. Let
$\lambda_{\max}\left(\bm{B}_{i}^{T}\bm{U}\right)$ be the largest
singular value of the matrix $\bm{B}_{i}^{T}\bm{U}$. Then \begin{equation}
f_{G,i}\left(\bm{U}\right)=1-\lambda_{\max}^{2}\left(\bm{B}_{i}^{T}\bm{U}\right).\label{eq:def-geometric-metric}\end{equation}

This expression is closely related to the chordal distance between
two subspaces, as described in Section \ref{sub:Grassmann-Manifold}.
We henceforth refer to the function (\ref{eq:def-geometric-metric})
either as the \emph{geometric metric} (\ref{eq:def-geometric-metric}),
or with slight abuse of terminology, as the chordal distance.

One advantage of the chordal distance is its continuity. This follows
directly from the continuity of the singular values of the underlying
matrix. Recall Example 1. In Fig. \ref{fig:Contours}, we illustrate
the differences between $f_{F,1}$ and $f_{G,1}$ by projecting their
contours of constant value onto the $u_{2}$-$u_{3}$ plane.

\begin{figure}
\includegraphics[scale=0.4]{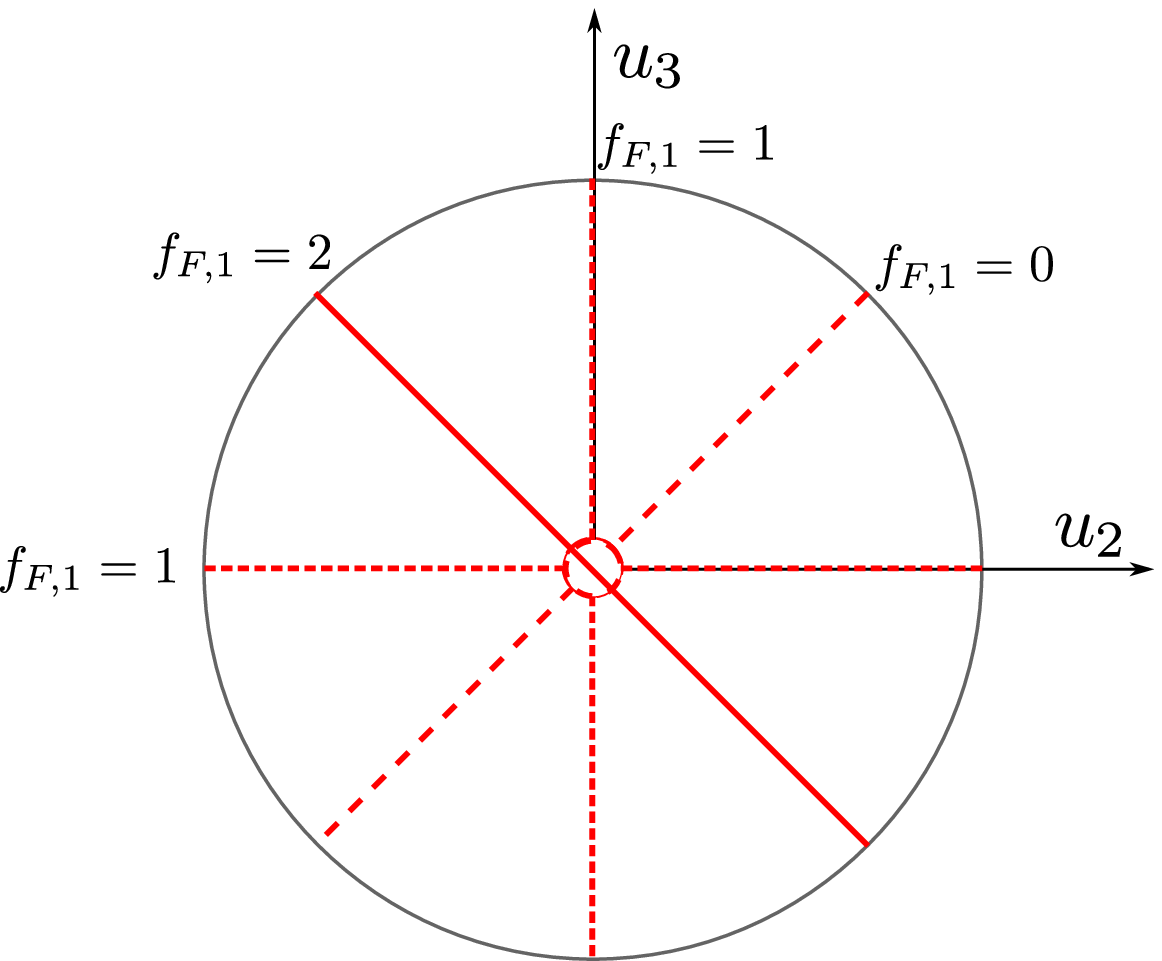}$\qquad$\includegraphics[scale=0.4]{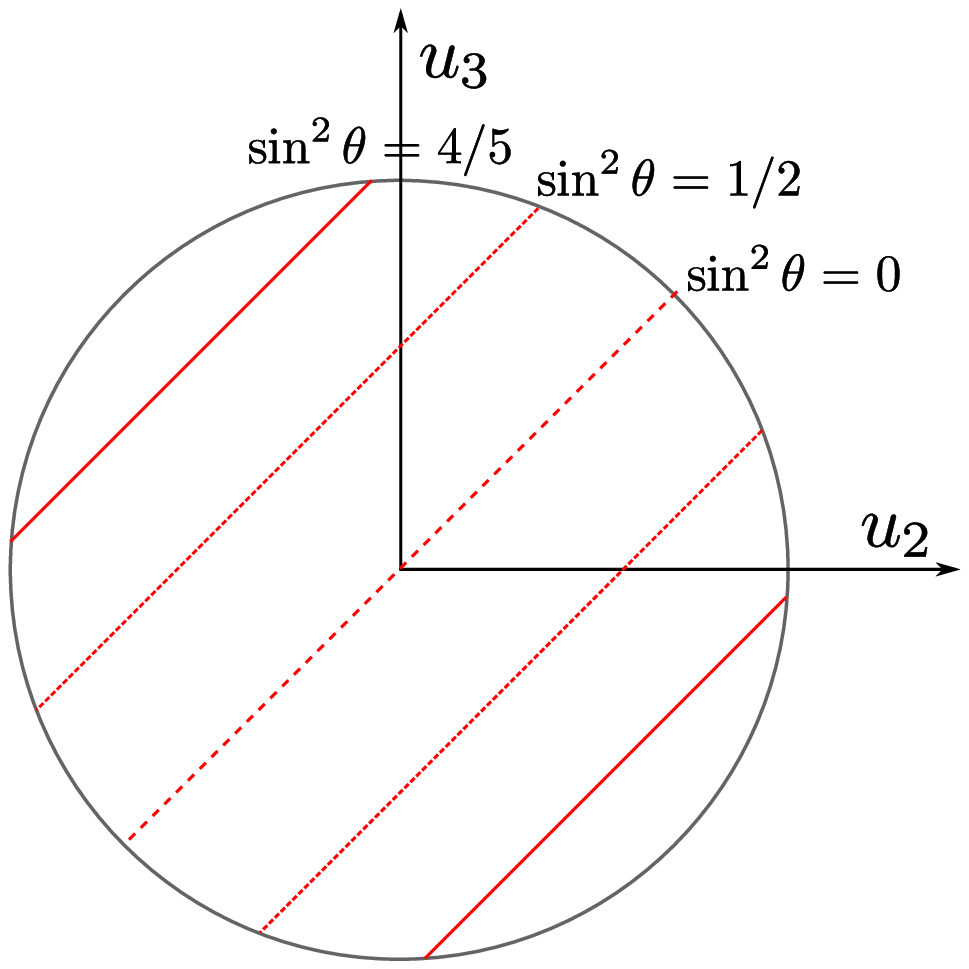}

\caption{\label{fig:Contours}Contours projected to the $\left(u_{2},u_{3}\right)$
plane. The left depicts the contours of the squared Frobenius norm.
The right corresponds to the chordal distance. }

\end{figure}

More importantly, the following theorem shows that the preimage of
$f_{G,i}\left(\bm{U}\right)=0$ is actually the closure of the preimage
of $f_{F,i}\left(\bm{U}\right)=0$. 
\begin{thm}
\label{thm:closure}Given $\bm{x}_{\Omega,i}\in\mathbb{R}^{m}$ and
$\Omega_{i}\subset\left[m\right]$. Let $\bm{U}_{\Omega_{i}}\in\mathbb{R}^{m\times r}$
be such that $\left(\bm{U}_{\Omega_{i}}\right)_{k,\ell}=\bm{U}_{k,\ell}$
if $k\in\Omega_{i}$ and $\left(\bm{U}_{\Omega_{i}}\right)_{k,\ell}=0$
if $k\notin\Omega_{i}$. Define \[
\mathcal{U}_{F,i}=\left\{ \bm{U}\in\mathcal{U}_{m,r}:\; f_{F,i}\left(\bm{U}\right)=\left\Vert \bm{x}_{\Omega,i}-\mathcal{P}\left(\bm{x}_{\Omega,i},\bm{U}_{\Omega_{i}}\right)\right\Vert ^{2}=0\right\} \]
 and\[
\mathcal{U}_{G,i}=\left\{ \bm{U}\in\mathcal{U}_{m,r}:\; f_{G,i}\left(\bm{U}\right)=1-\lambda_{\max}\left(\bm{B}_{i}^{T}\bm{U}\right)=0\right\} .\]
 Then $\mathcal{U}_{G,i}$ is the closure of $\mathcal{U}_{F,i}$,
i.e., $\mathcal{U}_{G,i}=\overline{\mathcal{U}_{F,i}}$. 
\end{thm}
\vspace{0.01in}

The proof is given in Appendix \ref{sub:pf-closure}. Although this
theorem deals with only one column of the observed matrix, the result
can be easily extended to the whole matrix $\bm{X}_{\Omega}$: let
$\mathcal{U}_{F}=\bigcap_{i=1}^{n}\mathcal{U}_{F,i}$ and \begin{align}
\mathcal{U}_{G} & =\bigcap_{i=1}^{n}\mathcal{U}_{G,i}\nonumber \\
 & =\left\{ \bm{U}\in\mathcal{U}_{m,r}:\;\lambda_{\max}\left(\bm{U}^{T}\bm{B}_{i}\right)=1\;\mbox{for all }i\right\} ;\label{eq:def-UG}\end{align}
 then $\mathcal{U}_{G}=\overline{\mathcal{U}_{F}}$.

\emph{Example 1 (Continued):} It can be seen that \[
\bm{B}_{1}=\left[\begin{array}{ccc}
0 & \frac{1}{\sqrt{2}} & \frac{1}{\sqrt{2}}\\
1 & 0 & 0\end{array}\right]^{T}.\]
Hence, \[
f_{G,1}\left(\bm{U}\right)=1-\lambda_{\max}^{2}\left(\left[\begin{array}{c}
\sqrt{2}\epsilon\\
\sqrt{1-2\epsilon^{2}}\end{array}\right]\right)=0.\]
As a result, \begin{align*}
\mathcal{U}_{F,1} & =\left\{ \left[\sqrt{1-2\epsilon^{2}},\epsilon,\epsilon\right]^{T}:\;\epsilon^{2}\le\frac{1}{2}\mbox{ and }\epsilon\ne0\right\} \\
 & \quad\bigcup\left\{ \left[-\sqrt{1-2\epsilon^{2}},\epsilon,\epsilon\right]^{T}:\;\epsilon^{2}\le\frac{1}{2}\mbox{ and }\epsilon\ne0\right\} ,\end{align*}
and \begin{align*}
\mathcal{U}_{G,1} & =\left\{ \left[\sqrt{1-2\epsilon^{2}},\epsilon,\epsilon\right]^{T}:\;\epsilon^{2}\le\frac{1}{2}\right\} \\
 & \quad\bigcup\left\{ \left[-\sqrt{1-2\epsilon^{2}},\epsilon,\epsilon\right]^{T}:\;\epsilon^{2}\le\frac{1}{2}\right\} .\end{align*}
Clearly, $\mathcal{U}_{G,1}=\overline{\mathcal{U}_{F,1}}$.

\vspace{0.03in}

\subsection{\label{sub:Computations-Chordal-Dist}Computations Related to the
Chordal Distance}

For a given performance metric, the computational complexity of the
supporting optimization procedure is an important factor for assessing
its practical value. In this subsection, we show that besides its
continuity, the chordal distance and the related gradient can be computed
efficiently. Hence, all the algorithmic solutions using gradient descent
methods can be easily modified to accommodate the geometric distortion
measure.

The principal angle $\theta_{i}$ and the chordal distance $\sin^{2}\theta_{i}$
can be computed using the singular value decomposition. Given the
$i^{th}$ column of the observed matrix, one can form $\bm{B}_{i}$
easily. Let $\lambda_{i}$ be the largest singular value of the matrix
$\bm{B}_{i}\bm{B}_{i}^{T}\bm{U}$, and let $\bm{b}_{i}$ and $\bm{v}_{i}$
be the corresponding left and right singular vectors, respectively%
\footnote{For convenience, we use the following convention regarding the singular
vectors $\bm{b}_{i}$ and $\bm{v}_{i}$: we let the first nonzero
entry of $\bm{v}_{i}$ be positive; otherwise, we let $\bm{v}_{i}^{\prime}=-\bm{v}_{i}$
and $\bm{b}_{i}^{\prime}=-\bm{b}_{i}$, and use $\bm{v}_{i}^{\prime}$
and $\bm{b}_{i}^{\prime}$ for singular value decomposition. The simultaneous
changes in signs do not affect the singular value decomposition nor
the computation of the gradient. %
}. Following the definition of the chordal distance, one has $f_{G,i}\left(\bm{U}\right)=\sin^{2}\theta_{i}=1-\lambda_{i}^{2}$.
Let $\bm{G}_{i}\in\mathbb{R}^{m\times r}$ be a matrix such that \[
\left(\bm{G}_{i}\right)_{k,\ell}=\frac{\partial}{\partial\bm{U}_{k,\ell}}f_{G,i}\left(\bm{U}\right)=-2\cos\theta_{i}\frac{\partial\cos\theta_{i}}{\partial\bm{U}_{k,\ell}}.\]
 It can be verified that \begin{equation}
\bm{G}_{i}=-2\lambda\bm{b}_{i}\bm{v}_{i}^{T}.\label{eq:Graident}\end{equation}
 Note that in the matrix completion problem, one only needs to search
for a column space $\mbox{span}\left(\bm{U}\right)$ consistent with
the observations. Taking this fact into consideration, we have \cite{edelman_optimization_manifolds_1998}
\begin{equation}
\nabla_{\bm{U}}f_{G}=\sum_{i=1}^{n}\nabla_{\bm{U}}f_{G,i}=\left(\bm{I}-\bm{U}\bm{U}^{T}\right)\sum_{i=1}^{n}\bm{G}_{i}.\label{eq:Gradient-G}\end{equation}

Switching from the Frobenius norm to the chordal distance does not
introduce extra computational cost. Due to the particular structure
of $\bm{B}_{i}$, the matrix multiplication $\bm{B}_{i}\bm{B}_{i}^{T}\bm{U}$
can be executed in $O\left(mr\right)$ steps. The resulting matrix
has dimensions $m\times r$, where typically $r\ll m$. The major
computational burden is incurred by the singular value decomposition.
Computing the largest singular value and the corresponding singular
vectors of an $m\times r$ matrix essentially reduces to computing
the largest eigenvalue of an $r\times r$ matrix and the corresponding
eigenvector. Hence, the overall complexity of computing $f_{G,i}$
is $O\left(mr^{2}+r^{3}\right)=O\left(mr^{2}\right)$, where the $O\left(mr^{2}\right)$
and $O\left(r^{3}\right)$ terms come from matrix multiplication and
eigenvalue computation, respectively. In comparison, to solve the
least square problem in the definition of $f_{F,i}$ has a $O\left(mr^{2}\right)$
cost as well.

\section{\label{sec:Performance-Guarantee}Performance Guarantees}

Consider the matrix completion problem described in (\ref{eq:matrix-completion-chordal-dist}).
The following theorem describes completion scenarios for which a global
optimum can be found with probability one.

\vspace{0.03in}

\begin{thm}
\label{thm:guarantee}Consider the following cases: 
\begin{enumerate}
\item \emph{(rank-one matrices with arbitrary sampling)}: Let $\bm{X}_{\Omega}=\mathcal{P}_{\Omega}\left(\bm{X}\right)$
for some unknown matrix $\bm{X}$ with rank equal to one. Here, $\Omega\subset\left[m\right]\times\left[n\right]$
can be arbitrary. 
\item \emph{(full sampling with arbitrary rank matrices)}: Let $\bm{X}_{\Omega}=\bm{X}$,
i.e., $\Omega=\left[m\right]\times\left[n\right]$. 
\end{enumerate}
Suppose that $r=\mbox{rank}\left(\bm{X}\right)$ is given. Let $\mathcal{U}_{G}\subset\mathcal{U}_{m,r}$
be the preimage of $f_{G}\left(\bm{U}\right)=0$ (also defined in
(\ref{eq:def-UG})). Let $\bm{U}_{0}$ be randomly generated from
the isotropic distribution on $\mathcal{U}_{m,r}$, and used as the
initial point of the search procedure. With probability one, there
exists a continuous path $\bm{U}\left(t\right)$, $t\in\left[0,1\right]$,
such that $\bm{U}\left(0\right)=\bm{U}_{0}$, $\bm{U}\left(1\right)\in\mathcal{U}_{G}$
and $\frac{d}{dt}f_{G}\le0$ for all $t\in\left(0,1\right)$, where
the equality holds if and only if $\bm{U}_{0}\in\mathcal{U}_{G}$.
\end{thm}
\vspace{0.03in}

The proof of the theorem is outlined in Section \ref{sub:pf-guarantee}.
It is worth to note that almost all starting points are good: it is
certainly good if the starting point is a consistent solution; otherwise,
there exists a continuous path from this starting point to a global
optimum such that the objective function keeps decreasing. The performance
guarantee provided in Theorem \ref{thm:guarantee} is strong in the
sense that it does not require either incoherence conditions or large
matrix sizes.

A simple corollay of the Theorem \ref{thm:guarantee} is the following
result: suppose that the partial observations $\bm{X}_{\Omega}$ admit
a unique consistent solution in terms of the Frobenius norm; then
a gradient search procedure using the geometric norm finds this unique
solution with probability one. This conclusion follows from the fact
that the solution set under the Frobenius norm contains only a single
point and therefore $\mathcal{U}_{G}=\overline{\mathcal{U}_{F}}=\mathcal{U}_{F}$. 

For the more general case where $r>1$ and $\Omega\ne\left[m\right]\times\left[n\right]$,
we can not prove the same performance guarantees. Nevertheless, in
Section \ref{sub:other-cases}, we present a collection of results
that may be helpful for future exploration.

\subsection{\label{sub:pf-guarantee}Proof of Theorem \ref{thm:guarantee}}

For our proof techniques, we need the following two assumptions.

\vspace{0.03in}

\textbf{\emph{Assumption I}}: There exists a global optimum $\bm{U}_{X}\in\mathcal{U}_{m,r}$
such that $f_{G}\left(\bm{U}_{\bm{X}}\right)=0$ and all the $r$
principal angles between $\mbox{span}\left(\bm{U}_{X}\right)$ and
$\mbox{span}\left(\bm{U}_{0}\right)$ are less than $\pi/2$. That
is, all the singular values of $\bm{U}_{X}^{T}\bm{U}_{0}$ are strictly
positive.

\textbf{\emph{Assumption II}}: All of the $\theta_{i}$'s (the smallest
principal angle between $\mbox{span}\left(\bm{U}_{0}\right)$ and
$\mbox{span}\left(\bm{B}_{i}\right)$) are less than $\pi/2$.

\vspace{0.03in}
 
\begin{remrk}
\begin{flushleft}
Suppose that the matrix $\bm{U}_{0}$ is randomly drawn from the uniform
(isotropic) distribution on $\mathcal{U}_{m,r}$. Then $\bm{U}_{0}$
satisfies both assumptions with probability one. This result can be
easily verified using the probability density function of the principal
angles \cite{James_54_Normal_Multivariate_Analysis_Orthogonal_Group,Adler_2004_Integrals_Grassmann,Dai2008_small_ball_quantization_GM,Dai_Globecom07_large_balls}. 
\par\end{flushleft}
\end{remrk}
\vspace{0.03in}

Assuming that these two assumptions are satisfied, we have the following
two theorems corresponding to the two cases in Theorem \ref{thm:guarantee},
respectively. 
\begin{thm}
\emph{\label{thm:guarantee-rank-one}(Rank-One Case)} Let $\bm{X}_{\Omega}$
be the partial observation matrix generated from a rank-one matrix.
Let $\bm{u}_{0}\in\mathcal{U}_{m,1}$ be an estimate of the column
space that satisfies Assumptions I and II. Suppose that $\sum_{i=1}^{n}\sin^{2}\theta_{i}\ne0$.
Then there exists a continuous path $\bm{u}\left(t\right)\in\mathcal{U}_{m,r}$
such that $\bm{u}\left(0\right)=\bm{u}_{0}$, $\bm{u}\left(1\right)\in\mathcal{U}_{G}$,
and $\left.\frac{d}{dt}\right|_{t=0}\sin^{2}\theta_{i}\le0$ for all
$i\in\left[n\right]$, where equality holds if and only if $\theta_{i}\left(0\right)=0$. 
\end{thm}
\vspace{0.03in}
 
\begin{thm}
\emph{\label{thm:guarantee-full-sampling}(Full-Sampling Case)} Let
$\bm{X}\in\mathbb{R}^{m\times n}$ be a rank-$r$ matrix. Let $\bm{U}_{0}\in\mathcal{U}_{m,r}$
satisfy Assumptions I and II. Suppose that $\sum_{i=1}^{n}\sin^{2}\theta_{i}\ne0$.
Then there exists a $\bm{U}\left(t\right)\in\mathcal{U}_{m,r}$ such
that $\bm{U}\left(0\right)=\bm{U}_{0}$, $\bm{U}\left(1\right)\in\mathcal{U}_{G}$
and $\left.\frac{d}{dt}\right|_{t=0}\sin^{2}\theta_{i}\le0$ for all
$i\in\left[n\right]$, where equality holds if and only if $\theta_{i}\left(0\right)=0$. 
\end{thm}
\vspace{0.03in}

The proofs of Theorem \ref{thm:guarantee-rank-one} and \ref{thm:guarantee-full-sampling}
are given in Appendix \ref{sub:pf-rank-1} and \ref{sub:pf-full-sampling},
respectively. Since the proof techniques differ significantly, we
present the two theorems/proofs separately.

Both theorems are stated for derivatives taken at $t=0$. Nevertheless,
the analysis can be extended for arbitrary $t\in\left[0,1\right]$,
that is, $\frac{d}{dt}\sin^{2}\theta_{i}\le0$ for all $t\in\left[0,1\right]$,
where the equality holds if and only if $\theta_{i}\left(t\right)=0$.
To show that this is the case, note that in proving both Theorem \ref{thm:guarantee-rank-one}
and Theorem \ref{thm:guarantee-full-sampling}, we constructed a continuous
path $\bm{U}\left(t\right)$ such that $\bm{U}\left(0\right)=\bm{U}_{0}$
and $\bm{U}\left(1\right)\in\mathcal{U}_{G}$. By fixing this continuous
path, we observe that: 
\begin{enumerate}
\item All the $r$ principal angles between $\mbox{span}\left(\bm{U}_{0}\right)$
and $\mbox{span}\left(\bm{U}\left(1\right)\right)$ are monotonically
decreasing as $t$ increases to one. This implies that Assumption
I holds for all $t\in\left[0,1\right]$. 
\item We have $\theta_{i}\left(t\right)<\pi/2$ for all $i\in\left[n\right]$
and for all $t\in\left[0,\epsilon\right)$ for some sufficiently small
$\epsilon>0$. This claim can be verified by invoking the facts that
$\theta_{i}\left(0\right)<\pi/2$ for all $i\in\left[n\right]$ and
that $\theta_{i}$ is a continuous functions for all $i\in\left[n\right]$.
As a result, all $\bm{U}\left(t\right)$'s, where $t\in\left[0,\epsilon\right)$,
satisfy Assumptions I and II. 
\item For every $t$ in the interval $\left[0,\epsilon\right)$, $\bm{U}\left(t\right)$
is the starting point of the geodesic path from $\bm{U}\left(t\right)$
to $\bm{U}\left(1\right)$, which is a part of the geodesic path from
$\bm{U}\left(0\right)$ to $\bm{U}\left(1\right)$. Using the same
proof techniques as in Appendix \ref{sub:pf-rank-1} and \ref{sub:pf-full-sampling},
it is clear that $\frac{d}{dt}\sin^{2}\theta_{i}\left(t\right)\le0$
for all $t\in\left[0,\epsilon\right)$. Hence, $\theta_{i}\left(t\right)\le\theta_{i}\left(0\right)<\frac{\pi}{2}$
for all $i\in\left[n\right]$ and for all $t\in\left[0,\epsilon\right)$. 
\item The arguments above can be extended. It can be verified that $\theta_{i}\left(t\right)\le\theta_{i}\left(0\right)<\pi/2$
for all $i\in\left[n\right]$ and for all $t\in\left[0,1\right]$.
This implies that $\bm{U}\left(t\right)$ satisfies Assumptions I
and II for all $t\in\left[0,1\right]$. Hence, $\frac{d}{dt}\sin^{2}\theta_{i}\left(t\right)\le0$
for all $i\in\left[n\right]$ and all $t\in\left[0,1\right]$, where
the equality holds if and only if $\theta_{i}\left(t\right)=0$. Theorem
\ref{thm:guarantee} therefore holds.
\end{enumerate}
\vspace{0.03in}

A direct consequence of Theorem \ref{thm:guarantee} is that for almost
all $\bm{U}_{0}\in\mathcal{U}_{m,r}$, there exists a continuous path
leading to a global minimizer. However, one does not know this path
in the process of solving the matrix completion problem. A practical
approach is to use a gradient descent method. We consider the following
randomized gradient descent algorithm. Let $\bm{U}^{\left(i\right)}\in\mathcal{U}_{m,r}$,
$i=1,2,\cdots$, be the starting point of the $i^{th}$ iteration.
Clearly, $\bm{U}^{\left(i\right)}$, $i\ge2$, is also the end point
of the $\left(i-1\right)^{th}$ iteration. We generate the sequence
of $\bm{U}^{\left(i\right)}$'s in the following manner. 
\begin{enumerate}
\item Let $\bm{U}^{\left(1\right)}$ be randomly generated from the isotropic
distribution. 
\item Set $i=1$. Execute the following iterative process.

\begin{enumerate}
\item Compute the gradient $\nabla_{\bm{U}^{\left(i\right)}}f_{G}$. 
\item Let $\bm{U}^{\left(i\right)}\left(t\right)$ be the geodesic curve
starting at $\bm{U}^{\left(i\right)}\left(0\right)=\bm{U}^{\left(i\right)}$
with direction $\bm{H}=-\nabla_{\bm{U}^{\left(i\right)}}f_{G}$. 
\item Let $t^{\left(i\right)*}$ be such that $\frac{d}{dt}f_{G}\left(t^{\left(i\right)*}\right)=0$
and $\frac{d}{dt}f_{G}\left(t\right)<0$ for all $t<t^{\left(i\right)*}$. 
\item Randomly generate a $t^{\left(i\right)}$ from the uniform distribution
on $\left(0,t^{\left(i\right)*}\right)$. 
\item Let $\bm{U}^{\left(i+1\right)}=\bm{U}^{\left(i\right)}\left(t^{\left(i\right)}\right)$.
Let $i=i+1$. Go to Step (a). 
\end{enumerate}
\end{enumerate}
Due to the randomness of $\bm{U}^{\left(i\right)}$, all $\bm{U}^{\left(i\right)}$'s
satisfy Assumptions I and II with probability one. The objective function
decreases after each iteration. This gradient descent procedure converges
to a global minimum as the number of iterations approachs infinity. 
\begin{remrk}
Denote the obtained global minimum by $\hat{\bm{U}}$. It may happen
that $\hat{\bm{U}}\in\mathcal{U}_{G}\backslash\mathcal{U}_{F}$. In
this case, the solution is inconsistent with respect to to the standard
Frobenius norm. One can use perturbation techniques to move $\hat{\bm{U}}$
from the boundary of $\mathcal{U}_{F}$ to the interior region of
$\mathcal{U}_{F}$. 
\end{remrk}

\subsection{\label{sub:other-cases}The General Framework}

For the cases that are not described in Theorem \ref{thm:guarantee},
we have the following corollary. 
\begin{cor}
\emph{\label{cor:General-Cases}(General Cases)} Let $\bm{X}\in\mathbb{R}^{m\times n}$
be a rank-$r$ matrix. Let $\bm{U}_{X}\in\mathcal{U}_{G}$ be a global
minimum. For each $i\in\left[n\right]$, the following statements
are true. Let $\bm{u}_{X,i}\in\mbox{span}\left(\bm{U}_{X}\right)\bigcap\mbox{span}\left(\bm{B}_{i}\right)$
be a unit norm vector. Let $\bm{U}_{0}\in\mathcal{U}_{m,r}$ and $\bm{w}_{i}\in\mathcal{U}_{r,1}$
be randomly drawn from the corresponding isotropic distributions respectively.
Then with probability one, the vector $\bm{u}_{0,i}\triangleq\bm{U}_{0}\bm{w}_{i}$
is not orthogonal to $\bm{u}_{X,i}$. Suppose that this is true. Define
$\theta_{i}=\cos^{-1}\left\Vert \mathcal{P}\left(\bm{u}_{i}\left(t\right),\bm{B}_{i}\right)\right\Vert _{2}$.
There exists a continuous path $\bm{u}_{i}\left(t\right)\in\mathcal{U}_{m,1}$
such that $\bm{u}_{i}\left(0\right)=\bm{u}_{0,i}$, $\bm{u}_{i}\left(1\right)\in\mbox{span}\left(\bm{U}_{X,i}\right)\bigcap\mathcal{U}_{m,1}$,
and $\frac{d}{dt}\sin^{2}\theta_{i}\le0$, where the equality holds
if and only if $\theta_{i}\left(t\right)=0$.\end{cor}
\begin{proof}
Without loss of generality, we assume that $\left\langle \bm{u}_{0,i},\bm{u}_{X,i}\right\rangle >0$.
The desired continuous path is given by \[
\bm{u}_{i}\left(t\right)=\frac{\left(1-t\right)\bm{u}_{0,i}+t\bm{u}_{X,i}}{\left\Vert \left(1-t\right)\bm{u}_{0,i}+t\bm{u}_{X,i}\right\Vert },\; t\in\left[0,1\right].\]
The detailed arguments are the same as those in the proof of Theorem
\ref{thm:guarantee-rank-one}, and therefore omitted. 
\end{proof}
\vspace{0.03in}
 
\begin{remrk}
This corollary is similar to Theorems \ref{thm:guarantee-rank-one}
and \ref{thm:guarantee-full-sampling} in the sense that there exist
continuous paths along which the atomic functions decreases. 

At the same time, Corollary \ref{cor:General-Cases} differs from
Theorems \ref{thm:guarantee-rank-one} and \ref{thm:guarantee-full-sampling}
in two aspects. First, the paths $\bm{u}_{i}\left(t\right)$ in Corollary
\ref{cor:General-Cases} may be different for different $i$'s, while
in Theorems \ref{thm:guarantee-rank-one} and \ref{thm:guarantee-full-sampling},
a single continuous path $\bm{U}\left(t\right)$ is constructed. Second,
the angle $\theta_{i}$ in Corollay \ref{cor:General-Cases} is essentially
the principal angle between the 1-dimensional subspace $\mbox{span}\left(\bm{u}_{i}\left(t\right)\right)$
and the subspace $\mbox{span}\left(\bm{B}_{i}\right)$. In contrast,
Theorem \ref{thm:guarantee-rank-one} and \ref{thm:guarantee-full-sampling}
involve the minimum principal angle between the $r$-dimensional subspace
$\mbox{span}\left(\bm{U}\left(t\right)\right)$ and the subspace $\mbox{span}\left(\bm{B}_{i}\right)$.
\end{remrk}

\section{\label{sec:Conclusion}Conclusion}

We considered the problem of how to search for a consistent completion
of low-rank matrices. We showed that Frobenius norm combined with
a projection operator results in a discontinuous objective function
and therefore makes gradient descent approach fail. We proposed to
replace the Frobenius norm with the chordal distance. The chordal
distance is the {}``best'' smooth version of the Frobenius norm
in the sense that the solution set of the former is the closure of
the solution set of the latter. Based on the chordal distance, we
derived strong performance guarantees for two completion scenarios.
The derived performance guarantees do not rely on incoherence conditions
or large matrix sizes, and they hold with probability one.

\appendix

\subsection{\label{sub:pf-closure}Proof of Theorem \ref{thm:closure}}

We omit the subscript $i$ to simplify notation. The proof consists
of two parts, showing that: 
\begin{enumerate}
\item $\mathcal{U}_{F}\subset\mathcal{U}_{G}$; 
\item for any given $\bm{U}_{0}\in\mathcal{U}_{G}$, there exists a sequence
$\left\{ \bm{U}^{\left(n\right)}\right\} \subset\mathcal{U}_{F}$
such that $\lim_{n\rightarrow\infty}\left\Vert \bm{U}_{0}-\bm{U}^{\left(n\right)}\right\Vert _{F}=0$. 
\end{enumerate}
\vspace{0in}

We start by proving that $\mathcal{U}_{F}\subset\mathcal{U}_{G}$.
For any given $\bm{U}\in\mathcal{U}_{F}$, there exists a nonzero
vector $\bm{w}\in\mathbb{R}^{r}$ such that $\bm{U}_{\Omega}\bm{w}=\bm{x}_{\Omega}$.
Let $\bm{b}=\bm{U}\bm{w}/\left\Vert \bm{w}\right\Vert $. Clearly,
$\left\Vert \bm{b}\right\Vert _{F}=1$. Recall the formula for $\bm{B}_{\bm{x}_{\Omega}}$.
We can write $\bm{b}$ as a linear combination of columns of $\bm{B}_{\bm{x}_{\Omega}}$:\begin{align*}
\bm{b} & =\frac{1}{\left\Vert \bm{w}\right\Vert }\bm{x}_{\Omega}+\sum_{j\in\Omega^{c}}b_{j}\bm{e}_{j}=\frac{\left\Vert \bm{x}_{\Omega}\right\Vert }{\left\Vert \bm{w}\right\Vert }\bar{\bm{x}}_{\Omega}+\sum_{j\in\Omega^{c}}b_{j}\bm{e}_{j}.\end{align*}
 As a result, \[
\left\Vert \bm{B}_{\bm{x}_{\Omega}}^{T}\bm{b}\right\Vert _{F}=\left\Vert \bm{B}_{\bm{x}_{\Omega}}^{T}\bm{U}\frac{\bm{w}}{\left\Vert \bm{w}\right\Vert _{F}}\right\Vert _{F}=1.\]
 It follows that the largest singular value of $\bm{B}_{\bm{x}_{\Omega}}^{T}\bm{U}$
is one. Therefore, $\bm{U}\in\mathcal{U}_{G}$, and we thus have $\mathcal{U}_{F}\subset\mathcal{U}_{G}$.

To prove the second part, we make use of the following notation. For
any given $\bm{U}_{0}\in\mathcal{U}_{G}$, let $\bm{u}_{1},\cdots,\bm{u}_{r}$
be the left singular vectors of the matrix $\bm{U}_{0}\bm{U}_{0}^{T}\bm{B}_{\bm{x}_{\Omega}}$
corresponding to the $i^{th}$ largest singular value. Let $k$ be
the multiplicity of the singular value one, i.e., the number of singular
values that equal to one. Let $\bm{U}_{1:k}=\left[\bm{u}_{1},\cdots,\bm{u}_{k}\right]$
and $\bm{U}_{k+1:r}=\left[\bm{u}_{k+1},\cdots,\bm{u}_{r}\right]$.
Clearly, $\lambda_{\max}\left(\bm{U}_{k+1:r}^{T}\bm{B}_{\bm{x}_{\Omega}}\right)<1$.

It suffices to focus on $\bm{U}$ instead of $\bm{U}_{0}$. That is,
to prove the second part, it suffices to find a sequence in $\mathcal{U}_{F}$
converging to $\bm{U}$. To verify this claim, let $\bm{V}=\bm{U}^{T}\bm{U}_{0}$.
Then $\bm{V}\in\mathcal{U}_{r,r}$ and $\bm{U}_{0}=\bm{U}\bm{V}$.
Suppose that $\left\{ \bm{U}^{\left(n\right)}\right\} \subset\mathcal{U}_{F}$
is a sequence such that $\bm{U}^{\left(n\right)}\rightarrow\bm{U}$.
It is clear that $\bm{U}^{\left(n\right)}\bm{V}\rightarrow\bm{U}\bm{V}=\bm{U}_{0}$.
Furthermore, since \[
\bm{x}_{\Omega}=\bm{U}_{\Omega}^{\left(n\right)}\bm{w}^{\left(n\right)}=\bm{U}_{\Omega}^{\left(n\right)}\bm{V}\left(\bm{V}^{T}\bm{w}^{\left(n\right)}\right)=\left(\bm{U}^{\left(n\right)}\bm{V}\right)_{\Omega}\bm{w}^{\prime\left(n\right)},\]
 one has $\bm{U}^{\left(n\right)}\bm{V}\in\mathcal{U}_{F}$. The sequence
$\left\{ \bm{U}^{\left(n\right)}\bm{V}\right\} \subset\mathcal{U}_{F}$
is the desired sequence that converges to $\bm{U}_{0}$. It is also
important to note that $\bm{U}\in\mathcal{U}_{G}$, since \[
\lambda\left(\bm{U}_{0}\bm{U}_{0}^{T}\bm{B}_{\bm{x}_{\Omega}}\right)=\lambda\left(\bm{U}\bm{V}\bm{V}^{T}\bm{U}^{T}\bm{B}_{\bm{x}_{\Omega}}\right)=\lambda\left(\bm{U}\bm{U}^{T}\bm{B}_{\bm{x}_{\Omega}}\right).\]

We claim that \begin{equation}
\bm{U}\in\mathcal{U}_{F}\;\mbox{if and only if }\bm{U}_{1:k,\Omega}\ne\bm{0}.\label{eq:U-UF}\end{equation}
To prove this claim, we shall show that \begin{equation}
\bm{U}_{1:k,\Omega}\ne\bm{0}\Rightarrow\bm{U}\in\mathcal{U}_{F}\label{eq:U-in-UF}\end{equation}
 and \begin{equation}
\bm{U}_{1:k,\Omega}=\bm{0}\Rightarrow\bm{U}\notin\mathcal{U}_{F}.\label{eq:U-notin-UF}\end{equation}

To prove (\ref{eq:U-in-UF}), suppose that $\bm{U}_{1:k,\Omega}\ne0$.
Without loss of generality, let $\bm{u}_{1,\Omega}\ne\bm{0}$. Since
$\bm{u}_{1}$ is the left singular vector corresponding to the singular
value equal to one, $\bm{u}_{1}$ can be written as a linear combination
of the columns of $\bm{B}_{\bm{x}_{\Omega}}$: $\bm{u}_{1}=a_{1}\bar{\bm{x}}_{\Omega}+\sum_{j\in\Omega^{c}}a_{j}\bm{e}_{j}$.
Since $\bm{u}_{1,\Omega}=a_{1}\bar{\bm{x}}_{\Omega}\ne\bm{0}$, one
has $a_{1}\ne0$. As a result, $\bm{x}_{\Omega}=a\bm{u}_{1,\Omega}$
for some constant $a\ne0$. Hence, $d_{F}\left(\bm{x}_{\Omega},\bm{U}\right)=0$
and $\bm{U}\in\mathcal{U}_{F}$.

To prove (\ref{eq:U-notin-UF}), assume that $\bm{U}_{1:k,\Omega}=\bm{0}$.
Since $\mathcal{P}\left(\bm{x}_{\Omega},\bm{U}_{\Omega}\right)=\mathcal{P}\left(\bm{x}_{\Omega},\bm{U}_{k+1:r,\Omega}\right)$,
proving that $\bm{U}\notin\mathcal{U}_{F}$ is equivalent to proving
that $\bm{x}_{\Omega}-\mathcal{P}\left(\bm{x}_{\Omega},\bm{U}_{k+1:r,\Omega}\right)\ne\bm{0}$.
This inequality can be proved by contradiction. Suppose that we have
an equality. Then there exists a vector $\bm{w}\in\mathbb{R}^{r-k}$
such that $\bm{U}_{k+1:r,\Omega}\bm{w}=\bm{x}_{\Omega}$. Let $\bm{b}=\bm{U}_{k+1:r}\bm{w}/\left\Vert \bm{w}\right\Vert $.
It is straightforward to show (using similar arguments as the ones
used for proving $\mathcal{U}_{F}\subset\mathcal{U}_{G}$) that $\bm{b}\in\mbox{span}\left(\bm{B}_{\bm{x}_{\Omega}}\right)$
and the largest singular value of $\bm{U}_{k+1:r}^{T}\bm{B}_{\bm{x}_{\Omega}}$
is one. This contradicts the fact that $\lambda_{\max}\left(\bm{U}_{k+1:r}^{T}\bm{B}_{\bm{x}_{\Omega}}\right)<1$.

Now we are ready to construct a sequence in $\mathcal{U}_{F}$ converging
to $\bm{U}$. If $\bm{U}_{1:k,\Omega}\ne0$, then $\bm{U}\in\mathcal{U}_{F}$
and it is trivial to find a sequence in $\mathcal{U}_{F}$ converging
to $\bm{U}$. It remains to find a sequence $\left\{ \bm{U}^{\left(n\right)}\right\} \subset\mathcal{U}_{F}$
that converges to $\bm{U}$ when $\bm{U}_{1:k,\Omega}=0$. Define
$\bm{x}_{r}=\bm{x}_{\Omega}-\mathcal{P}\left(\bm{x}_{\Omega},\bm{U}_{\Omega}\right)$.
Since $\bm{U}_{1:k,\Omega}=0$, one has $\bm{U}\notin\mathcal{U}_{F}$
and $\bm{x}_{r}\ne\bm{0}$. Note that $\bm{x}_{r,\Omega^{c}}=\bm{0}$
and that $\bm{x}_{r,\Omega}\perp\bm{u}_{i,\Omega}$ for all $i\in\left[r\right]$.
It can be verified that $\bm{x}_{r}\perp\bm{u}_{1}$, $\cdots$, $\bm{x}_{r}\perp\bm{u}_{r}$.
Let \[
\bm{U}_{\epsilon}=\left[\frac{\bm{u}_{1}+\epsilon\bm{x}_{r}}{\sqrt{1+\epsilon^{2}}},\bm{u}_{2},\cdots,\bm{u}_{r}\right].\]
 It can be verified that $\bm{U}_{\epsilon}\in\mathcal{U}_{m,r}$.
Furthermore, $\mathcal{P}\left(\bm{x}_{\Omega,}\bm{U}_{\Omega}\right)=\mathcal{P}\left(\bm{x}_{\Omega},\left[\bm{x}_{r},\bm{U}_{k+1:r,\Omega}\right]\right)=\bm{x}_{\Omega}$
and therefore $\bm{U}_{\epsilon}\in\mathcal{U}_{F}$ for all $\epsilon\ne0$.
Now choose a sequence $\left\{ \bm{U}^{\left(n\right)}\right\} =\left\{ \bm{U}_{1/n}\right\} $.
It is a sequence in $\mathcal{U}_{F}$ and it converges to $\bm{U}$.
This completes the proof.

\subsection{\label{sub:pf-rank-1}Proof of Theorem \ref{thm:guarantee-rank-one}}

Since $\bm{X}_{\Omega}$ is generated from a rank-one matrix, there
exists a $\bm{u}_{\bm{X}}\in\mathcal{U}_{m,1}$ such that $\bm{u}_{\bm{X}}\in\mbox{span}\left(\bm{B}_{i}\right)$
for all $i\in\left[n\right]$. Without loss of generality, we assume
$\left\langle \bm{u},\bm{u}_{\bm{X}}\right\rangle >0$: by Assumption
I, $\left\langle \bm{u},\bm{u}_{\bm{X}}\right\rangle \ne0$; if $\left\langle \bm{u},\bm{u}_{\bm{X}}\right\rangle <0$,
we replace $\bm{u}_{\bm{X}}$ with $-\bm{u}_{\bm{X}}$.

Now define \[
\bm{u}\left(t\right)=\frac{\left(1-t\right)\bm{u}_{0}+t\bm{u}_{\bm{X}}}{\left\Vert \left(1-t\right)\bm{u}_{0}+t\bm{u}_{\bm{X}}\right\Vert }=\frac{\left(1-t\right)\bm{u}_{0}+t\bm{u}_{\bm{X}}}{L\left(t\right)},\]
 where $L\left(t\right)\triangleq\left\Vert \left(1-t\right)\bm{u}_{0}+t\bm{u}_{\bm{X}}\right\Vert .$
Clearly $\bm{u}\left(0\right)=\bm{u}_{0}$ and $\bm{u}\left(t\right)\in\mathcal{U}_{m,1}$
in a neighborhood of $t=0$.

For every $i\in\left[n\right]$, we shall show that \begin{equation}
\left.\frac{d}{dt}\right|_{t=0}\sin^{2}\theta_{i}=-2\left.\frac{d}{dt}\right|_{t=0}\left(\frac{1}{2}\cos^{2}\theta_{i}\right)\le0,\label{eq:pf-rank1-dsin2}\end{equation}
 where the equality holds if and only if $\theta_{i}=0$. Let $\mathcal{P}_{i}\bm{u}$
denote the vector $\mathcal{P}\left(\bm{u},\bm{B}_{i}\right)=\bm{B}_{i}\bm{B}_{i}^{T}\bm{u}$.
Since $\bm{u}_{\bm{X}}\in\mbox{span}\left(\bm{B}_{i}\right)$, one
has \[
\mathcal{P}_{i}\bm{u}=\frac{1}{L\left(t\right)}\left(\left(1-t\right)\mathcal{P}_{i}\bm{u}_{0}+t\bm{u}_{\bm{X}}\right).\]
 We then have\begin{align*}
 & \left.\frac{d}{dt}\right|_{t=0}\left(\frac{1}{2}\cos^{2}\theta_{i}\right)=\left.\frac{d}{dt}\right|_{t=0}\frac{1}{2}\left\Vert \mathcal{P}_{i}\bm{u}\right\Vert ^{2}\\
 & \quad=\left.\frac{d}{dt}\right|_{t=0}\left[\frac{1}{2}\left(\frac{1-t}{L\left(t\right)}\right)^{2}\left\Vert \mathcal{P}_{i}\bm{u}_{0}\right\Vert ^{2}\right.\\
 & \quad\quad\left.+\frac{1}{2}\left(\frac{t}{L\left(t\right)}\right)^{2}+\frac{\left(t-t^{2}\right)}{L^{2}\left(t\right)}\left\langle \mathcal{P}_{i}\bm{u}_{0},\bm{u}_{\bm{X}}\right\rangle \right]\\
 & \quad=\left(-1-L^{\prime}\left(0\right)\right)\left\Vert \mathcal{P}_{i}\bm{u}_{0}\right\Vert ^{2}+\left\langle \mathcal{P}_{i}\bm{u}_{0},\bm{u}_{\bm{X}}\right\rangle .\end{align*}
 Note that\[
\left\langle \mathcal{P}_{i}\bm{u}_{0},\bm{u}_{\bm{X}}\right\rangle =\bm{u}_{\bm{X}}^{T}\bm{B}_{i}\bm{B}_{i}^{T}\bm{u}_{0}=\left\langle \bm{u}_{0},\mathcal{P}_{i}\bm{u}_{\bm{X}}\right\rangle =\left\langle \bm{u}_{0},\bm{u}_{\bm{X}}\right\rangle .\]
Consequently, \begin{equation}
\left.\frac{d}{dt}\right|_{t=0}\left(\frac{1}{2}\cos^{2}\theta_{i}\right)=\left(-1-L^{\prime}\left(0\right)\right)\left\Vert \mathcal{P}_{i}\bm{u}_{0}\right\Vert ^{2}+\left\langle \bm{u}_{0},\bm{u}_{\bm{X}}\right\rangle .\label{eq:pf-rank1-dcos2}\end{equation}
 The term $L^{\prime}\left(0\right)$ can be computed as follows.
Note that \begin{align*}
L^{2}\left(t\right) & =\left(1-t\right)^{2}\left\Vert \bm{u}_{0}\right\Vert ^{2}+t^{2}\left\Vert \bm{u}_{\bm{X}}\right\Vert +2\left(t-t^{2}\right)\left\langle \bm{u}_{0},\bm{u}_{\bm{X}}\right\rangle \\
 & =1-2t+2t^{2}+2\left(t-t^{2}\right)\left\langle \bm{u}_{0},\bm{u}_{\bm{X}}\right\rangle .\end{align*}
 Therefore, \begin{align*}
\left.\frac{d}{dt}\right|_{t=0}L^{2}\left(t\right) & =-2+2\left\langle \bm{u}_{0},\bm{u}_{\bm{X}}\right\rangle =2L\left(0\right)L^{\prime}\left(0\right).\end{align*}
As a result, \begin{equation}
L^{\prime}\left(0\right)=-1+\left\langle \bm{u}_{0},\bm{u}_{\bm{X}}\right\rangle .\label{eq:pf-rank1-dL}\end{equation}
 Substituting (\ref{eq:pf-rank1-dL}) into (\ref{eq:pf-rank1-dcos2})
one can see that \[
\left.\frac{d}{dt}\right|_{t=0}\left(\frac{1}{2}\cos^{2}\theta_{i}\right)=\left\langle \bm{u}_{0},\bm{u}_{\bm{X}}\right\rangle \left(1-\left\Vert \mathcal{P}_{i}\bm{u}_{0}\right\Vert ^{2}\right)\ge0,\]
 where the equality holds if and only if $\left\Vert \mathcal{P}_{i}\bm{u}_{0}\right\Vert =1$,
i.e., $\bm{u}_{0}\in\mbox{span}\left(\bm{B}_{i}\right)$ and $\theta_{i}=0$.
This completes the proof.

\subsection{\label{sub:pf-full-sampling}Proof of Theorem \ref{thm:guarantee-full-sampling}}

Let $\bm{U}_{\bm{X}}\in\mathcal{U}_{m,r}$ be such that every column
of $\bm{X}$ is in the subspace $\mbox{span}\left(\bm{U}_{\bm{X}}\right)$.
Consider the compact singular decomposition $\bm{U}_{0}\bm{U}_{0}^{T}\bm{U}_{\bm{X}}\bm{U}_{\bm{X}}^{T}=\bm{U}_{0}^{\prime}\bm{S}\bm{U}_{\bm{X}}^{\prime T}$,
where $\bm{S}\in\mathbb{R}^{r\times r}$ is the diagonal matrix containing
the singular values and $\bm{U}_{0}^{\prime}$ and $\bm{U}_{\bm{X}}^{\prime}$
are the left and right singular vector matrices, respectively. Clearly,
$\bm{U}_{0}$ and $\bm{U}_{0}^{\prime}$ generate the same subspace,
and so do $\bm{U}_{\bm{X}}$ and $\bm{U}_{\bm{X}}^{\prime}$. For
simplicity, we present our proof for $\bm{U}_{0}^{\prime}$ and $\bm{U}_{\bm{X}}^{\prime}$
and omit the superscripts. With this simplification, one has $\bm{U}^{T}\bm{U}_{\bm{X}}=\bm{S}=\mbox{diag}\left(\left[\lambda_{1},\cdots,\lambda_{r}\right]\right)$.

For the $i^{th}$ column of $\bm{X}$, we compute $\nabla_{\bm{U}_{0}}\cos\theta_{i}$.
Since we are considering the full sampling case, we have $\bm{B}_{i}=\bar{\bm{x}}_{i}$.
Because $\bar{\bm{x}}_{i}\in\mbox{span}\left(\bm{U}_{\bm{X}}\right)$,
there exists $\bar{\bm{w}}\in\mathcal{U}_{r,1}$ such that $\bar{\bm{x}}_{i}=\bm{U}_{\bm{X}}\bar{\bm{w}}$.
To compute $\nabla_{\bm{U}_{0}}\cos\theta_{i}$, we need the first
left and the first right singular vectors of the matrix $\bar{\bm{x}}_{i}\bar{\bm{x}}_{i}^{T}\bm{U}_{0}$.
The first left singular vector is clearly $\bar{\bm{x}}_{i}$ and
the first right singular vector equals  $\bm{U}_{0}^{T}\bar{\bm{x}}_{i}=\bm{U}_{0}^{T}\bm{U}_{\bm{X}}\bar{\bm{w}}=\bm{S}\bar{\bm{w}}$.
Hence, \begin{align*}
\nabla_{\bm{U}_{0}}\cos\theta_{i} & =\left(\bm{I}-\bm{U}_{0}\bm{U}_{0}^{T}\right)\bar{\bm{x}}_{i}\bar{\bm{w}}^{T}\bm{S}^{T}\\
 & =\left(\bm{I}-\bm{U}_{0}\bm{U}_{0}^{T}\right)\bm{U}_{\bm{X}}\bar{\bm{w}}\bar{\bm{w}}^{T}\bm{S}^{T}.\end{align*}
 According to Lemma \ref{lem:Grassmann-Manifold}, $\left(\bm{I}-\bm{U}_{0}\bm{U}_{0}^{T}\right)\bm{U}_{\bm{X}}$
can be written as $\bm{G}\mbox{diag}\left(\left[\sin\alpha_{1},\cdots,\sin\alpha_{j}\right]\right)$,
where $\bm{G}=\left[\bm{g}_{1},\cdots,\bm{g}_{r}\right]\in\mathcal{U}_{m,r}$,
and $\alpha_{i}=\cos^{-1}\lambda_{i}$'s, $i=1,\cdots,r$, are the
principal angles between $\mbox{span}\left(\bm{U}_{0}\right)$ and
$\mbox{span}\left(\bm{U}_{\bm{X}}\right)$.

We consider the geodesic $\bm{U}\left(t\right)$ from $\bm{U}_{0}$
to $\bm{U}_{\bm{X}}$. In Lemma \ref{lem:Grassmann-Manifold} (part
1), we show that this geodesic is given by the $\bm{U}\left(t\right)$
satisfying $\bm{U}\left(0\right)=\bm{U}_{0}$ and $\dot{\bm{U}}\left(0\right)=\bm{G}\mbox{diag}\left(\left[\alpha_{1},\cdots,\alpha_{r}\right]\right)$.
Along this path, we have \begin{align}
 & \left.\frac{d}{dt}\right|_{t=0}\cos\theta_{i}=\left\langle \nabla_{\bm{U}_{0}}\cos\theta_{i},\bm{G}\mbox{diag}\left(\left[\alpha_{1},\cdots,\alpha_{r}\right]\right)\right\rangle \nonumber \\
 & =\mbox{trace}\left(\left(\bm{G}\mbox{diag}\left(\left[\alpha_{1},\cdots,\alpha_{r}\right]\right)\right)^{T}\right.\nonumber \\
 & \quad\left.\left(\left(\bm{I}-\bm{U}_{0}\bm{U}_{0}^{T}\right)\bm{U}_{\bm{X}}\right)\bar{\bm{w}}\bar{\bm{w}}^{T}\bm{S}^{T}\right)\nonumber \\
 & =\mbox{trace}\left(\mbox{diag}\left(\left[\cdots,\alpha_{j}\sin\alpha_{j},\cdots\right]\right)\bar{\bm{w}}\bar{\bm{w}}^{T}\bm{S}\right)\nonumber \\
 & =\mbox{trace}\left(\left(\bm{I}-\bm{S}^{2}\right)\bar{\bm{w}}\bar{\bm{w}}^{T}\bm{S}\right)\nonumber \\
 & =\sum_{j=1}^{r}\bar{w}_{j}^{2}\alpha_{j}\sin\alpha_{j}\cos\alpha_{j}\ge0.\label{eq:pf-fullsampling-dcos}\end{align}

We claim that under Assumption II, equality in (\ref{eq:pf-fullsampling-dcos})
holds if and only if $\theta_{i}=0$. If $\theta_{i}=0$, then $\bar{\bm{x}}_{i}\in\mbox{span}\left(\bm{U}_{0}\right)$.
According to Lemma \ref{lem:Grassmann-Manifold} (part 2), $\bar{w}_{j}=0$
for all $j$ such that $\alpha_{j}\ne0$. The equality in (\ref{eq:pf-fullsampling-dcos})
thus holds. Otherwise, if $\theta_{i}\ne0$, then $\bar{\bm{x}}_{i}\notin\mbox{span}\left(\bm{U}_{0}\right)$.
Again, according to Lemma \ref{lem:Grassmann-Manifold} (part 2),
there exists an $j\in\left[r\right]$ such that $\alpha_{i}>0$ and
$\bar{w}_{j}\ne0$. Hence, we have a strict inequality in (\ref{eq:pf-fullsampling-dcos}).
Finally, note that \[
\left.\frac{d}{dt}\right|_{t=0}\sin^{2}\theta_{i}=-2\left.\frac{d}{dt}\right|_{t=0}\cos\theta_{i}\le0.\]
 This proves the theorem. 

\bibliographystyle{ieeetr}
\bibliography{MatrixCompletion}

\end{document}